\documentclass[a4paper,11pt]{article}
\usepackage{amsfonts}
\usepackage{amsmath}
\usepackage{amssymb}
\usepackage{geometry}

\newtheorem{proposition}{Proposition}

\newenvironment{proof}[1][Proof]{\noindent\textbf{#1.} }{\ \rule{0.5em}{0.5em}}
\input{tcilatex}

\begin{document}

\title{Spinors and gravity without Lorentz indices}
\author{John Fredsted\thanks{%
physics@johnfredsted.dk} \\
R\o m\o v\ae nget 32B, 8381 Tilst, Denmark}
\maketitle

\begin{abstract}
Coupling spinor fields to the gravitational field, in the setting of general
relativity, is standardly done via the introduction of a vierbein field and
the (associated minimal) spin connection field. This makes three types of
indices feature in the formalism: world/coordinate indices, Lorentz vector
indices, and Lorentz spinor indices, respectively. This article will show,
though, that it is possible to dispense altogther with the Lorentz indices,
both tensorial ones and spinorial ones, obtaining a formalism featuring only
world indices. This will be possible by having both the 'Dirac operator' and
the generators of 'Lorentz' transformations become spacetime-dependent,
although covariantly constant. The formalism is developed in the setting of
complexified quaternions.
\end{abstract}

\section{Introduction}

According to standard wisdom, see for instance \cite[Sec. 31.A]{Weinberg} or 
\cite[Sec. 12.1]{GSW}, spinor fields can be coupled to the gravitational
field, in the setting of general relativity, only via the introduction of
fields carrying Lorentz vector indices (in excess of world indices), more
specifically, the vierbein field $e^{\mu }{}_{a}$, and the spin connection
field $\omega _{\mu }{}^{ab}$. In this way, the resulting formalism ends up
featuring, somewhat unsatisfactorily, three different types of indices:
world indices, Lorentz vector indices, and Lorentz spinor indices,
respectively, the latter of course being carried by the spinor field itself.

This article will present a formalism, though, using world indices only,
i.e., a formalism in which no Lorentz indices feature, neither tensorial
ones nor spinorial ones. Perhaps surprisingly, it will prove possible to
have the spinor field carry a world index (transforming as such under
coordinate transformations), rather than a Lorentz spinor index, while as a
field still transforming in the standard spinor representation of the
Lorentz group in any local Lorentz frame. This will be achieved by having
both the 'Dirac operator' and the generators of 'Lorentz' transformations
become spacetime-dependent, although covariantly constant. By carrying a
world index, the spinor field may then readily be coupled to the
gravitational field via the connection field $\Gamma ^{\rho }{}_{\mu \nu }$
by which tensorial fields are coupled, thus implementing, it would seem, the
equivalence principle in a more coherent way than in the standard vierbein
formalism. The structure needed to set up the formalism will be constructed
in terms of quantities valued in the complexified quaternions.

\section{\label{Sec:Preliminaries}Preliminaries}

The set of complexified quaternions is denoted $\mathbb{C}\otimes \mathbb{H}$%
, equal to $\mathbb{H}\otimes \mathbb{C}$ as any two elements from $\mathbb{C%
}$ and $\mathbb{H}$, respectively, are assumed to (multiplicatively)
commute: $ch=hc$, for all $\left( c,h\right) \in \mathbb{C}\times \mathbb{H}$%
. Usual complex conjugation, $x\rightarrow x^{\ast }$, is assumed to act
only on $\mathbb{C}$, and usual quaternionic conjugation, $x\rightarrow 
\overline{x}$, is assumed to act only on $\mathbb{H}$. More specifically,%
\begin{eqnarray*}
\left( \mathbb{C}\otimes \mathbb{H}\right) ^{\ast } &=&\left\{ c^{\ast
}h\left\vert c\in \mathbb{C},h\in \mathbb{H}\right. \right\} , \\
\overline{\left( \mathbb{C}\otimes \mathbb{H}\right) } &=&\left\{ c\overline{%
h}\left\vert c\in \mathbb{C},h\in \mathbb{H}\right. \right\} .
\end{eqnarray*}%
In conjunction, these two conjugations can be used to split $\mathbb{C}%
\otimes \mathbb{H}$ as $\mathbb{C}\otimes \mathbb{H}=\left( \mathbb{C}%
\otimes \mathbb{H}\right) ^{+}\cup \left( \mathbb{C}\otimes \mathbb{H}%
\right) ^{-}$, where%
\begin{eqnarray*}
\left( \mathbb{C}\otimes \mathbb{H}\right) ^{+} &\equiv &\left\{ x\in 
\mathbb{C}\otimes \mathbb{H}\left\vert \overline{x}^{\ast }=+x\right.
\right\} , \\
\left( \mathbb{C}\otimes \mathbb{H}\right) ^{-} &\equiv &\left\{ x\in 
\mathbb{C}\otimes \mathbb{H}\left\vert \overline{x}^{\ast }=-x\right.
\right\} ,
\end{eqnarray*}%
this being an almost disjoint union, zero being the only common element of $%
\left( \mathbb{C}\otimes \mathbb{H}\right) ^{+}$ and $\left( \mathbb{C}%
\otimes \mathbb{H}\right) ^{-}$. The scalar- and vector parts of $\mathbb{C}%
\otimes \mathbb{H}$, respectively, are denoted $\mathrm{Scal}\left( \mathbb{C%
}\otimes \mathbb{H}\right) \cong \mathbb{C}$ and $\mathrm{Vec}\left( \mathbb{%
C}\otimes \mathbb{H}\right) =\left( \mathbb{C}\otimes \mathbb{H}\right)
\setminus \mathrm{Scal}\left( \mathbb{C\otimes \mathbb{H}}\right) $. A
bilinear inner product $\left\langle \cdot ,\cdot \right\rangle :\left( 
\mathbb{C}\otimes \mathbb{H}\right) ^{2}\rightarrow \mathbb{C}$ is given by%
\begin{eqnarray}
2\left\langle x,y\right\rangle &\equiv &x\overline{y}+y\overline{x}  \notag
\\
&\equiv &\overline{x}y+\overline{y}x.  \label{Eq:IPDef}
\end{eqnarray}%
The literature disagrees on the presence or not of the factor of $2$, this
however being inconsequential as long as the same factor is consistently
used throughout; for instance, if the factor figures in Eq. (\ref{Eq:IPDef})
above, then it will have to figure as well in Eq. (\ref{Eq:IPxyuv}) below.
As $\mathbb{C}\otimes \mathbb{H}$ is a socalled composition algebra \cite%
{SpringerAndVeldkamp,Okubo}, in fact a somewhat dull one as it is
associative, this inner product satisfies the following relations, among
other ones:%
\begin{eqnarray}
\left\langle x,y\right\rangle &=&\left\langle y,x\right\rangle ,
\label{Eq:IPSym} \\
\left\langle x,y\right\rangle &=&\left\langle \overline{x},\overline{y}%
\right\rangle ,  \label{Eq:IPConj} \\
\left\langle xy,z\right\rangle &=&\left\langle y,\overline{x}z\right\rangle
=\left\langle x,z\overline{y}\right\rangle ,  \label{Eq:IP(xy,z)} \\
\left\langle x,yz\right\rangle &=&\left\langle \overline{y}x,z\right\rangle
=\left\langle x\overline{z},y\right\rangle ,  \label{Eq:IP(x,yz)}
\end{eqnarray}%
not all independent but listed nonetheless for completeness, and%
\begin{equation}
\left\langle xu,yv\right\rangle +\left\langle xv,yu\right\rangle
=2\left\langle x,y\right\rangle \left\langle u,v\right\rangle ,
\label{Eq:IPxyuv}
\end{equation}%
for any $x,y,z,u,v\in \mathbb{C}\otimes \mathbb{H}$. As a basis (over $%
\mathbb{C}$) for $\mathbb{C}\otimes \mathbb{H}$, any four elements $q_{\mu
}\in \mathbb{C}\otimes \mathbb{H}$ for which $\det \left( \left\langle
q_{\mu },q_{\nu }\right\rangle \right) \neq 0$, will suffice, as then $%
\mathbb{C}\otimes \mathbb{H}=\mathrm{Span}_{\mathbb{C}}\left( q_{\mu
}\right) $. But a more specific choice of basis will be made: Let $s_{\mu
}\in \left( \mathbb{C}\otimes \mathbb{H}\right) ^{+}$ for which $\det \left(
\left\langle s_{\mu },s_{\nu }\right\rangle \right) \neq 0$. Then $%
\left\langle s_{\mu },s_{\nu }\right\rangle $ as a matrix will be symmetric,
due to Eq. (\ref{Eq:IPSym}); real-valued, due to%
\begin{eqnarray*}
\left\langle s_{\mu },s_{\nu }\right\rangle ^{\ast } &=&\left\langle s_{\mu
}^{\ast },s_{\nu }^{\ast }\right\rangle \\
&=&\left\langle \overline{s}_{\mu },\overline{s}_{\nu }\right\rangle \\
&=&\left\langle s_{\mu },s_{\nu }\right\rangle ,
\end{eqnarray*}%
using $s_{\mu }^{\ast }=\overline{s}_{\mu }$ and Eq. (\ref{Eq:IPConj}); and
non-singular, due to the determinantal condition. Furthermore, it has
signature $\left( 1,3\right) $, i.e., if it is diagonalized, then one
diagonal element will be positive, and three diagonal elements will be
negative; this is readily seen from the specific example $s_{\mu }=\left( 1,%
\mathrm{i}e_{i}\right) $, where $e_{i}\in \mathrm{Vec}\left( \mathbb{H}%
\right) $ are the standard quaternionic units obeying $e_{i}e_{j}=-\delta
_{ij}+\varepsilon _{ijk}e_{k}$. Having these properties, it is natural to
identify this quantity with the metric of a signature $\left( 1,3\right) $
Riemann-Cartan spacetime:%
\begin{equation}
g_{\mu \nu }\equiv \left\langle s_{\mu },s_{\nu }\right\rangle ,
\label{Eq:MetricDef}
\end{equation}%
thus, at the same time, elevating $s_{\mu }$ to a type $\left( 0,1\right) $
tensor field. The corresponding type $\left( 1,0\right) $ tensor field $%
s^{\mu }$ is then, of course, given by $s^{\mu }=g^{\mu \nu }s_{\nu }$. With
these two types of fields at hand, the following completeness relation may
be shown to hold:%
\begin{equation}
\left\langle x,y\right\rangle =\left\langle x,s_{\mu }\right\rangle
\left\langle s^{\mu },y\right\rangle \equiv g^{\mu \nu }\left\langle
x,s_{\mu }\right\rangle \left\langle s_{\nu },y\right\rangle ,
\label{Eq:CompletenessRelation}
\end{equation}%
for any $x,y\in \mathbb{C}\otimes \mathbb{H}$.

\section{'Modified Clifford algebra'}

Consider the following complex-valued type $\left( 2,1\right) $ tensor field:%
\begin{equation}
M^{\mu \rho }{}_{\sigma }\equiv \left\langle s^{\mu },s^{\rho }s_{\sigma
}\right\rangle ,  \label{Eq:MDef}
\end{equation}%
not to be confused with any connection field (which is not even a tensor
field, of course). Under complex conjugation, it behaves as follows:%
\begin{eqnarray}
\left( M^{\mu \rho }{}_{\sigma }\right) ^{\ast } &=&\left\langle \overline{s}%
^{\mu },\overline{s}^{\rho }\overline{s}_{\sigma }\right\rangle  \notag \\
&=&\left\langle s^{\mu },s_{\sigma }s^{\rho }\right\rangle  \notag \\
&=&M^{\mu }{}_{\sigma }{}^{\rho },  \label{Eq:MComplexConj}
\end{eqnarray}%
using $s_{\mu }^{\ast }=\overline{s}_{\mu }$ and Eq. (\ref{Eq:IPConj}).

\begin{proposition}
\label{Prop:MAlgebra}$M^{\mu \rho }{}_{\sigma }$ satisfies the following
algebra:%
\begin{eqnarray}
2g^{\mu \nu }\delta _{\sigma }^{\rho } &=&M^{\mu \rho }{}_{\tau }\left(
M^{\nu \tau }{}_{\sigma }\right) ^{\ast }+M^{\nu \rho }{}_{\tau }\left(
M^{\mu \tau }{}_{\sigma }\right) ^{\ast }  \label{Eq:MMStar} \\
&=&\left( M^{\mu \rho }{}_{\tau }\right) ^{\ast }M^{\nu \tau }{}_{\sigma
}+\left( M^{\nu \rho }{}_{\tau }\right) ^{\ast }M^{\mu \tau }{}_{\sigma }.
\label{Eq:MStarM}
\end{eqnarray}
\end{proposition}

\begin{proof}
As the metric is real-valued, the assertion implied by the second equality
follows immediately from the assertion of the first line. It is thus
sufficient to prove the latter, say. By direct calculation:%
\begin{eqnarray*}
M^{\mu \rho }{}_{\tau }\left( M^{\nu \tau }{}_{\sigma }\right) ^{\ast
}+M^{\nu \rho }{}_{\tau }\left( M^{\mu \tau }{}_{\sigma }\right) ^{\ast }
&=&\left\langle s^{\mu },s^{\rho }s_{\tau }\right\rangle \left\langle 
\overline{s}^{\nu },\overline{s}^{\tau }\overline{s}_{\sigma }\right\rangle
+\left\langle s^{\nu },s^{\rho }s_{\tau }\right\rangle \left\langle 
\overline{s}^{\mu },\overline{s}^{\tau }\overline{s}_{\sigma }\right\rangle
\\
&=&\left\langle \overline{s}^{\rho }s^{\mu },s_{\tau }\right\rangle
\left\langle \overline{s}^{\nu }s_{\sigma },\overline{s}^{\tau
}\right\rangle +\left\langle \overline{s}^{\rho }s^{\nu },s_{\tau
}\right\rangle \left\langle \overline{s}^{\mu }s_{\sigma },\overline{s}%
^{\tau }\right\rangle \\
&=&\left\langle \overline{s}^{\rho }s^{\mu },s_{\tau }\right\rangle
\left\langle s^{\tau },\overline{s}_{\sigma }s^{\nu }\right\rangle
+\left\langle \overline{s}^{\rho }s^{\nu },s_{\tau }\right\rangle
\left\langle s^{\tau },\overline{s}_{\sigma }s^{\mu }\right\rangle \\
&=&\left\langle \overline{s}^{\rho }s^{\mu },\overline{s}_{\sigma }s^{\nu
}\right\rangle +\left\langle \overline{s}^{\rho }s^{\nu },\overline{s}%
_{\sigma }s^{\mu }\right\rangle \\
&=&2\left\langle \overline{s}^{\rho },\overline{s}_{\sigma }\right\rangle
\left\langle s^{\nu },s^{\mu }\right\rangle \\
&=&2\left\langle s^{\mu },s^{\nu }\right\rangle \left\langle s^{\rho
},s_{\sigma }\right\rangle \\
&=&2g^{\mu \nu }\delta _{\sigma }^{\rho },
\end{eqnarray*}%
using several of the properties of the inner product listed in Sec. \ref%
{Sec:Preliminaries}.
\end{proof}

In terms of $4\times 4$ matrices $\mathbf{M}^{\mu }$ with components $\left( 
\mathbf{M}^{\mu }\right) ^{\rho }{}_{\sigma }\equiv M^{\mu \rho }{}_{\sigma
} $, this algebra may also be written concisely in matrix notation as%
\begin{eqnarray}
2g^{\mu \nu }\mathbf{1} &=&\mathbf{M}^{\mu }\mathbf{M}^{\nu \ast }+\mathbf{M}%
^{\nu }\mathbf{M}^{\mu \ast }  \label{Eq:MMStar_Matrix} \\
&=&\mathbf{M}^{\mu \ast }\mathbf{M}^{\nu }+\mathbf{M}^{\nu \ast }\mathbf{M}%
^{\mu },  \label{Eq:MStarM_Matrix}
\end{eqnarray}%
where $\mathbf{1}$ is the identity matrix. Apart from the complex
conjugations, this algebra is the $Cl\left( 1,3\right) $ Clifford algebra,
and it may thus perhaps be called a 'modified Clifford algebra' (the algebra
may certainly have been studied somewhere in the literature, and thus have a
specific name, but the author is not aware of any such). The relevance of
this algebra will become clear shortly.

\section{'Modified Lorentz algebra'}

Consider the following complex-valued type $\left( 3,1\right) $ tensor field:%
\begin{eqnarray}
4S^{\mu \nu \rho }{}_{\sigma } &\equiv &\left\langle \overline{s}^{\mu
}s^{\rho },\overline{s}^{\nu }s_{\sigma }\right\rangle -\left\langle 
\overline{s}^{\nu }s^{\rho },\overline{s}^{\mu }s_{\sigma }\right\rangle
\label{Eq:SDef} \\
&\equiv &\left\langle s^{\rho },\left( s^{\mu }\overline{s}^{\nu }-s^{\nu }%
\overline{s}^{\mu }\right) s_{\sigma }\right\rangle ,  \notag
\end{eqnarray}

\begin{proposition}
\label{Prop:SAlgebra}$S^{\mu \nu \rho }{}_{\sigma }$ may be written in terms
of $M^{\mu }{}_{\rho \sigma }$ as follows:%
\begin{equation}
4S^{\mu \nu \rho }{}_{\sigma }=M^{\mu \rho }{}_{\tau }\left( M^{\nu \tau
}{}_{\sigma }\right) ^{\ast }-M^{\nu \rho }{}_{\tau }\left( M^{\mu \tau
}{}_{\sigma }\right) ^{\ast }.  \label{Eq:SMM}
\end{equation}
\end{proposition}

\begin{proof}
By direct calculation:%
\begin{eqnarray*}
M^{\mu \rho }{}_{\tau }\left( M^{\nu \tau }{}_{\sigma }\right) ^{\ast
}-M^{\nu \rho }{}_{\tau }\left( M^{\mu \tau }{}_{\sigma }\right) ^{\ast }
&=&\left\langle s^{\mu },s^{\rho }s_{\tau }\right\rangle \left\langle 
\overline{s}^{\nu },\overline{s}^{\tau }\overline{s}_{\sigma }\right\rangle
-\left\langle s^{\nu },s^{\rho }s_{\tau }\right\rangle \left\langle 
\overline{s}^{\mu },\overline{s}^{\tau }\overline{s}_{\sigma }\right\rangle
\\
&=&\left\langle \overline{s}^{\rho }s^{\mu },s_{\tau }\right\rangle
\left\langle \overline{s}^{\nu }s_{\sigma },\overline{s}^{\tau
}\right\rangle -\left\langle \overline{s}^{\rho }s^{\nu },s_{\tau
}\right\rangle \left\langle \overline{s}^{\mu }s_{\sigma },\overline{s}%
^{\tau }\right\rangle \\
&=&\left\langle \overline{s}^{\rho }s^{\mu },s_{\tau }\right\rangle
\left\langle s^{\tau },\overline{s}_{\sigma }s^{\nu }\right\rangle
-\left\langle \overline{s}^{\rho }s^{\nu },s_{\tau }\right\rangle
\left\langle s^{\tau },\overline{s}_{\sigma }s^{\mu }\right\rangle \\
&=&\left\langle \overline{s}^{\rho }s^{\mu },\overline{s}_{\sigma }s^{\nu
}\right\rangle -\left\langle \overline{s}^{\rho }s^{\nu },\overline{s}%
_{\sigma }s^{\mu }\right\rangle \\
&=&\left\langle \overline{s}^{\mu }s^{\rho },\overline{s}^{\nu }s_{\sigma
}\right\rangle -\left\langle \overline{s}^{\nu }s^{\rho },\overline{s}^{\mu
}s_{\sigma }\right\rangle ,
\end{eqnarray*}%
most of the steps being analogous to the ones taken in the proof of
Proposition \ref{Prop:MAlgebra}.
\end{proof}

In terms of $4\times 4$ matrices $\mathbf{S}^{\mu \nu }$ with components $%
\left( \mathbf{S}^{\mu \nu }\right) ^{\rho }{}_{\sigma }\equiv S^{\mu \nu
\rho }{}_{\sigma }$, Eq. (\ref{Eq:SMM}) may also be written concisely in
matrix notation as%
\begin{equation}
4\mathbf{S}^{\mu \nu }=\mathbf{M}^{\mu }\mathbf{M}^{\nu \ast }-\mathbf{M}%
^{\nu }\mathbf{M}^{\mu \ast },  \label{Eq:SMM_Matrix}
\end{equation}%
using as well the previously defined matrices $\mathbf{M}^{\mu }$.

\begin{proposition}
In conjunction, $\mathbf{S}^{\mu \nu }$ and $\mathbf{M}^{\rho }$ satisfy the
following identity:%
\begin{equation}
\mathbf{M}^{\rho }\mathbf{S}^{\mu \nu \ast }-\mathbf{S}^{\mu \nu }\mathbf{M}%
^{\rho }=g^{\rho \mu }\mathbf{M}^{\nu }-g^{\rho \nu }\mathbf{M}^{\mu }.
\label{Eq:DoubleCover}
\end{equation}%
This is in the present formalism the analogue of the identity $\left[ \gamma
^{\rho },S^{\mu \nu }\right] =\left( V^{\mu \nu }\right) ^{\rho }{}_{\sigma
}\gamma ^{\sigma }$ from the standard Dirac formalism \cite{Peskin and
Schroeder}.
\end{proposition}

\begin{proof}
By direct calculation:%
\begin{eqnarray*}
4\mathbf{S}^{\mu \nu }\mathbf{M}^{\rho } &\equiv &\left( \mathbf{M}^{\mu }%
\mathbf{M}^{\nu \ast }-\mathbf{M}^{\nu }\mathbf{M}^{\mu \ast }\right) 
\mathbf{M}^{\rho } \\
&=&\mathbf{M}^{\mu }\left( -\mathbf{M}^{\rho \ast }\mathbf{M}^{\nu
}+2g^{\rho \nu }\mathbf{1}\right) -\mathbf{M}^{\nu }\left( -\mathbf{M}^{\rho
\ast }\mathbf{M}^{\mu }+2g^{\rho \mu }\mathbf{1}\right) \\
&=&-\mathbf{M}^{\mu }\mathbf{M}^{\rho \ast }\mathbf{M}^{\nu }+\mathbf{M}%
^{\nu }\mathbf{M}^{\rho \ast }\mathbf{M}^{\mu }-2g^{\rho \mu }\mathbf{M}%
^{\nu }+2g^{\rho \nu }\mathbf{M}^{\mu } \\
&=&-\left( -\mathbf{M}^{\rho }\mathbf{M}^{\mu \ast }+2g^{\rho \mu }\mathbf{1}%
\right) \mathbf{M}^{\nu }+\left( -\mathbf{M}^{\rho }\mathbf{M}^{\nu \ast
}+2g^{\rho \nu }\mathbf{1}\right) \mathbf{M}^{\mu }-2g^{\rho \mu }\mathbf{M}%
^{\nu }+2g^{\rho \nu }\mathbf{M}^{\mu } \\
&=&\mathbf{M}^{\rho }\left( \mathbf{M}^{\mu \ast }\mathbf{M}^{\nu }-\mathbf{M%
}^{\nu \ast }\mathbf{M}^{\mu }\right) -4g^{\rho \mu }\mathbf{M}^{\nu
}+4g^{\rho \nu }\mathbf{M}^{\mu } \\
&\equiv &4\mathbf{M}^{\rho }\mathbf{S}^{\mu \nu \ast }-4g^{\rho \mu }\mathbf{%
M}^{\nu }+4g^{\rho \nu }\mathbf{M}^{\mu },
\end{eqnarray*}%
using Eqs. (\ref{Eq:MMStar_Matrix})-(\ref{Eq:MStarM_Matrix}) and (\ref%
{Eq:SMM_Matrix}).
\end{proof}

\begin{proposition}
The matrices $\mathbf{S}^{\mu \nu }$ satisfy the following algebra:%
\begin{equation}
\left[ \mathbf{S}^{\mu \nu },\mathbf{S}^{\rho \sigma }\right] =-\left(
g^{\mu \rho }\mathbf{S}^{\nu \sigma }-g^{\mu \sigma }\mathbf{S}^{\nu \rho
}-g^{\nu \rho }\mathbf{S}^{\mu \sigma }+g^{\nu \sigma }\mathbf{S}^{\mu \rho
}\right) ,  \label{Eq:ModLorentzAlgebra}
\end{equation}%
this being the Lorentz algebra with $\eta ^{\mu \nu }$ replaced by $g^{\mu
\nu }$ (it may thus perhaps be called a 'modified Lorentz algebra').
\end{proposition}

\begin{proof}
By direct calculation: From%
\begin{eqnarray*}
\mathbf{S}^{\mu \nu }\mathbf{M}^{\rho }\mathbf{M}^{\sigma \ast } &=&\left[ 
\mathbf{M}^{\rho }\mathbf{S}^{\mu \nu \ast }-\left( g^{\mu \rho }\mathbf{M}%
^{\nu }-g^{\nu \rho }\mathbf{M}^{\mu }\right) \right] \mathbf{M}^{\sigma
\ast } \\
&=&\mathbf{M}^{\rho }\mathbf{S}^{\mu \nu \ast }\mathbf{M}^{\sigma \ast
}-g^{\mu \rho }\mathbf{M}^{\nu }\mathbf{M}^{\sigma \ast }+g^{\nu \rho }%
\mathbf{M}^{\mu }\mathbf{M}^{\sigma \ast } \\
&=&\mathbf{M}^{\rho }\left[ \mathbf{M}^{\sigma \ast }\mathbf{S}^{\mu \nu
}-\left( g^{\mu \sigma }\mathbf{M}^{\nu \ast }-g^{\nu \sigma }\mathbf{M}%
^{\mu \ast }\right) \right] -g^{\mu \rho }\mathbf{M}^{\nu }\mathbf{M}%
^{\sigma \ast }+g^{\nu \rho }\mathbf{M}^{\mu }\mathbf{M}^{\sigma \ast } \\
&=&\mathbf{M}^{\rho }\mathbf{M}^{\sigma \ast }\mathbf{S}^{\mu \nu }-g^{\mu
\sigma }\mathbf{M}^{\rho }\mathbf{M}^{\nu \ast }+g^{\nu \sigma }\mathbf{M}%
^{\rho }\mathbf{M}^{\mu \ast }-g^{\mu \rho }\mathbf{M}^{\nu }\mathbf{M}%
^{\sigma \ast }+g^{\nu \rho }\mathbf{M}^{\mu }\mathbf{M}^{\sigma \ast },
\end{eqnarray*}%
using Eq. (\ref{Eq:DoubleCover}) and its complex conjugate, follows%
\begin{eqnarray*}
4\mathbf{S}^{\mu \nu }\mathbf{S}^{\rho \sigma } &\equiv &\mathbf{S}^{\mu \nu
}\left( \mathbf{M}^{\rho }\mathbf{M}^{\sigma \ast }-\mathbf{M}^{\sigma }%
\mathbf{M}^{\rho \ast }\right) \\
&=&\left( \mathbf{M}^{\rho }\mathbf{M}^{\sigma \ast }\mathbf{S}^{\mu \nu
}-g^{\mu \sigma }\mathbf{M}^{\rho }\mathbf{M}^{\nu \ast }+g^{\nu \sigma }%
\mathbf{M}^{\rho }\mathbf{M}^{\mu \ast }-g^{\mu \rho }\mathbf{M}^{\nu }%
\mathbf{M}^{\sigma \ast }+g^{\nu \rho }\mathbf{M}^{\mu }\mathbf{M}^{\sigma
\ast }\right) \\
&&-\left( \mathbf{M}^{\sigma }\mathbf{M}^{\rho \ast }\mathbf{S}^{\mu \nu
}-g^{\mu \rho }\mathbf{M}^{\sigma }\mathbf{M}^{\nu \ast }+g^{\nu \rho }%
\mathbf{M}^{\sigma }\mathbf{M}^{\mu \ast }-g^{\mu \sigma }\mathbf{M}^{\nu }%
\mathbf{M}^{\rho \ast }+g^{\nu \sigma }\mathbf{M}^{\mu }\mathbf{M}^{\rho
\ast }\right) \\
&=&\left( \mathbf{M}^{\rho }\mathbf{M}^{\sigma \ast }-\mathbf{M}^{\sigma }%
\mathbf{M}^{\rho \ast }\right) \mathbf{S}^{\mu \nu }-g^{\mu \rho }\left( 
\mathbf{M}^{\nu }\mathbf{M}^{\sigma \ast }-\mathbf{M}^{\sigma }\mathbf{M}%
^{\nu \ast }\right) +g^{\mu \sigma }\left( \mathbf{M}^{\nu }\mathbf{M}^{\rho
\ast }-\mathbf{M}^{\rho }\mathbf{M}^{\nu \ast }\right) \\
&&+g^{\nu \rho }\left( \mathbf{M}^{\mu }\mathbf{M}^{\sigma \ast }-\mathbf{M}%
^{\sigma }\mathbf{M}^{\mu \ast }\right) -g^{\nu \sigma }\left( \mathbf{M}%
^{\mu }\mathbf{M}^{\rho \ast }-\mathbf{M}^{\rho }\mathbf{M}^{\mu \ast
}\right) \\
&=&4\mathbf{S}^{\rho \sigma }\mathbf{S}^{\mu \nu }-4\left( g^{\mu \rho }%
\mathbf{S}^{\nu \sigma }-g^{\mu \sigma }\mathbf{S}^{\nu \rho }-g^{\nu \rho }%
\mathbf{S}^{\mu \sigma }+g^{\nu \sigma }\mathbf{S}^{\mu \rho }\right) .
\end{eqnarray*}%
\noindent This proof is structurally quite analogous to the proof in the
standard Dirac formalism that $S^{\mu \nu }=\frac{1}{4}\left[ \gamma ^{\mu
},\gamma ^{\nu }\right] $ satisfies the standard Lorentz algebra.
\end{proof}

Consider as well the following type $\left( 3,1\right) $ tensor field:%
\begin{eqnarray}
V^{\mu \nu \rho }{}_{\sigma } &\equiv &\left\langle s^{\mu },s^{\rho
}\right\rangle \left\langle s^{\nu },s_{\sigma }\right\rangle -\left\langle
s^{\nu },s^{\rho }\right\rangle \left\langle s^{\mu },s_{\sigma
}\right\rangle  \notag \\
&=&g^{\mu \rho }\delta _{\sigma }^{\nu }-g^{\nu \rho }\delta _{\sigma }^{\mu
}.  \label{Eq:VDef}
\end{eqnarray}%
The corresponding $4\times 4$ matrices $\mathbf{V}^{\mu \nu }$ with
components $\left( \mathbf{V}^{\mu \nu }\right) ^{\rho }{}_{\sigma }\equiv
V^{\mu \nu \rho }{}_{\sigma }$ are readily shown to satisfy the very same
algebra as do $\mathbf{S}^{\mu \nu }$, Eq. (\ref{Eq:ModLorentzAlgebra}).

Note that in a local inertial frame in which $\left\langle s^{\mu },s^{\nu
}\right\rangle =\eta ^{\mu \nu }$, the algebra Eq. (\ref%
{Eq:ModLorentzAlgebra}), with spacetime-dependent structure constants,
reduces to the standard Lorentz algebra, with spacetime-independent
structure constants. So in a local inertial frame, where $g^{\mu \nu }=\eta
^{\mu \nu }$, the quantities $\mathbf{S}^{\mu \nu }$ and $\mathbf{V}^{\mu
\nu }$ are ordinary representations of the standard (un-modified) Lorentz
algebra. More specifically, $\mathbf{S}^{\mu \nu }$ is a (spin $\frac{1}{2}$%
) spinor representation, and $\mathbf{V}^{\mu \nu }$ is a vector
representation, these assertions being readily established by calculating
their corresponding Casimir operators, $\frac{1}{2}\mathbf{S}^{ij}\mathbf{S}%
_{ij}$ and $\frac{1}{2}\mathbf{V}^{ij}\mathbf{V}_{ij}$, respectively, for
the $\mathrm{SO}\left( 3\right) $ subgroup of the Lorentz group. Note that
in this article, a notation without explicit $\mathrm{i}$'s in the
definition of the generators and, correspondingly, in the Lie algebra is
used.

\begin{proposition}
$\left( \mathbf{S}^{\mu \nu }\right) ^{\rho }{}_{\sigma }$ and $\left( 
\mathbf{V}^{\mu \nu }\right) ^{\rho }{}_{\sigma }$ satisfy the following
identities:%
\begin{eqnarray}
g_{\sigma \alpha }\left( \mathbf{S}^{\mu \nu }\right) ^{\alpha }{}_{\beta
}g^{\beta \rho } &=&-\left( \mathbf{S}^{\mu \nu }\right) ^{\rho }{}_{\sigma
},  \label{Eq:gSg} \\
g_{\sigma \alpha }\left( \mathbf{V}^{\mu \nu }\right) ^{\alpha }{}_{\beta
}g^{\beta \rho } &=&-\left( \mathbf{V}^{\mu \nu }\right) ^{\rho }{}_{\sigma
}.  \label{Eq:gVg}
\end{eqnarray}
\end{proposition}

\begin{proof}
As Eq. (\ref{Eq:gVg}) is just the generalization to curvilinear coordinates
of the well-known identity from special relativity responsible for the
invariance of the line element, only the proof of Eq. (\ref{Eq:gSg}) will be
given. By direct calculation:%
\begin{eqnarray*}
4g_{\sigma \alpha }\left( \mathbf{S}^{\mu \nu }\right) ^{\alpha }{}_{\beta
}g^{\beta \rho } &=&g_{\sigma \alpha }\left\langle s^{\alpha },\left( s^{\mu
}\overline{s}^{\nu }-s^{\nu }\overline{s}^{\mu }\right) s_{\beta
}\right\rangle g^{\beta \rho } \\
&=&\left\langle s_{\sigma },\left( s^{\mu }\overline{s}^{\nu }-s^{\nu }%
\overline{s}^{\mu }\right) s^{\rho }\right\rangle \\
&=&-\left\langle \left( s^{\mu }\overline{s}^{\nu }-s^{\nu }\overline{s}%
^{\mu }\right) s_{\sigma },s^{\rho }\right\rangle \\
&=&-\left\langle s^{\rho },\left( s^{\mu }\overline{s}^{\nu }-s^{\nu }%
\overline{s}^{\mu }\right) s_{\sigma }\right\rangle \\
&=&-\left( \mathbf{S}^{\mu \nu }\right) ^{\rho }{}_{\sigma },
\end{eqnarray*}%
using Eqs. (\ref{Eq:IPSym}) and (\ref{Eq:IP(x,yz)}).
\end{proof}

\begin{proposition}
The following identity holds:%
\begin{eqnarray}
2\left( \mathbf{S}^{\mu \nu }\right) ^{\rho \sigma } &=&g^{\mu \rho }g^{\nu
\sigma }-g^{\nu \rho }g^{\mu \sigma }\pm \mathrm{i}\varepsilon ^{\mu \nu
\rho \sigma }  \notag \\
&\equiv &\left( \mathbf{V}^{\mu \nu }\right) ^{\rho \sigma }\pm \mathrm{i}%
\varepsilon ^{\mu \nu \rho \sigma },  \label{Eq:SVLeviCivita}
\end{eqnarray}%
for either the plus or the minus sign. Here, $\varepsilon ^{\mu \nu \rho
\sigma }$ is the Levi-Civita tensor of type $\left( 4,0\right) $ \cite[p. 202%
]{MTW}.
\end{proposition}

It was this identity that originally inspired the investigations resulting
in the findings reported in this article. A sign change of the Levi-Civita
tensor term corresponds to the (parity) transformation $s_{\mu }\rightarrow 
\overline{s}_{\mu }=s_{\mu }^{\ast }$, which does not change the metric and
thus neither $\left( \mathbf{V}^{\mu \nu }\right) ^{\rho \sigma }$ nor $%
\varepsilon ^{\mu \nu \rho \sigma }$ (the latter dedending only on the
determinant of the metric), but changes $\left( \mathbf{S}^{\mu \nu }\right)
^{\rho \sigma }$ into its complex conjugate. The choice of 'handedness' of $%
s_{\mu }$ is freely disposable, but once taken, it must of course be
consistently adhered to throughout.

Before giving the proof of Eq. (\ref{Eq:SVLeviCivita}), first a little
preparation, introducing some auxiliary machinery: Consider the map $%
X:\left( \mathbb{C}\otimes \mathbb{H}\right) ^{3}\rightarrow \mathbb{C}%
\otimes \mathbb{H}$ given by%
\begin{equation}
3!X\left( x,y,z\right) \equiv \left( x\overline{y}-y\overline{x}\right)
z+\left( y\overline{z}-z\overline{y}\right) x+\left( z\overline{x}-x%
\overline{z}\right) y,  \label{Eq:XDef}
\end{equation}%
which by construction is completely antisymmetric in its three arguments. An
aside: It satisfies the following orthogonality property and (generalized)
Pythagorean property, respectively:%
\begin{eqnarray*}
0 &=&\left\langle X\left( x_{1},x_{2},x_{3}\right) ,x_{i}\right\rangle , \\
\det \left( \left\langle x_{i},x_{j}\right\rangle \right) &=&\left\langle
X\left( x_{1},x_{2},x_{3}\right) ,X\left( x_{1},x_{2},x_{3}\right)
\right\rangle ,
\end{eqnarray*}%
for any $x_{i}\in \mathbb{C}\otimes \mathbb{H}$, where $i=1,2,3$, of course.
It corresponds to the complexification of a triple cross product on $\mathbb{%
R}^{4}$, taking three vectors and producing a single vector. Such cross
products with three factors, possessing these two properties, are possible
only in $\mathbb{R}^{4}$ and $\mathbb{R}^{8}$ \cite[Sec. 7.5]{Lounesto}, and
by complexification in $\mathbb{C}^{4}$ and $\mathbb{C}^{8}$, the underlying
reason being the existence of the division algebras $\mathbb{H}$ and $%
\mathbb{O}$.

The expanded right-hand side of Eq. (\ref{Eq:XDef}) consists of six terms,
of course. By using $2\left\langle x,y\right\rangle \equiv x\overline{y}+y%
\overline{x}\equiv \overline{x}y+\overline{y}x$, etc., repeatedly to
rearrange all terms as the first term (plus some extra) yields the following
equivalent expression:%
\begin{equation}
X\left( x,y,z\right) =x\overline{y}z-\left\langle y,z\right\rangle
x+\left\langle z,x\right\rangle y-\left\langle x,y\right\rangle z.
\label{Eq:XAlternative}
\end{equation}%
From it follows that (note interchange of $z$ and $u$ in the second relation
as compared to the first)%
\begin{eqnarray*}
\left\langle X\left( x,y,z\right) ,u\right\rangle &=&\left\langle x\overline{%
y}z,u\right\rangle -\left\langle y,z\right\rangle \left\langle
x,u\right\rangle +\left\langle z,x\right\rangle \left\langle
y,u\right\rangle -\left\langle x,y\right\rangle \left\langle
z,u\right\rangle , \\
\left\langle X\left( x,y,u\right) ,z\right\rangle &=&\left\langle x\overline{%
y}u,z\right\rangle -\left\langle y,u\right\rangle \left\langle
x,z\right\rangle +\left\langle u,x\right\rangle \left\langle
y,z\right\rangle -\left\langle x,y\right\rangle \left\langle
u,z\right\rangle ,
\end{eqnarray*}%
which added yields%
\begin{eqnarray*}
\left\langle X\left( x,y,z\right) ,u\right\rangle +\left\langle X\left(
x,y,u\right) ,z\right\rangle &=&\left\langle x\overline{y}z,u\right\rangle
+\left\langle x\overline{y}u,z\right\rangle -2\left\langle x,y\right\rangle
\left\langle z,u\right\rangle \\
&=&\left\langle x\overline{y},u\overline{z}+z\overline{u}\right\rangle
-2\left\langle x,y\right\rangle \left\langle z,u\right\rangle \\
&=&2\left\langle x\overline{y},1\right\rangle \left\langle z,u\right\rangle
-2\left\langle x,y\right\rangle \left\langle z,u\right\rangle \\
&=&2\left\langle x,y\right\rangle \left\langle z,u\right\rangle
-2\left\langle x,y\right\rangle \left\langle z,u\right\rangle \\
&\equiv &0,
\end{eqnarray*}%
using Eqs. (\ref{Eq:IPDef}) and (\ref{Eq:IP(xy,z)}), showing that $%
\left\langle X\left( x,y,z\right) ,u\right\rangle $ is antisymmetric in its
last two arguments, $z$ and $u$. As it is as well completely antisymmetric
in its first three arguments, due to the complete antisymmetry of $X\left(
x,y,z\right) $, as previously noted, it is in fact completely antisymmetric
in all its arguments. Now to the proof of Eq. (\ref{Eq:SVLeviCivita}):

\begin{proof}
From Eqs. (\ref{Eq:XAlternative}) follows, using Eq. (\ref{Eq:MetricDef}),%
\begin{eqnarray*}
X\left( s^{\mu },s^{\nu },s^{\rho }\right) &=&s^{\mu }\overline{s}^{\nu
}s^{\rho }-g^{\nu \rho }s^{\mu }+g^{\rho \mu }s^{\nu }-g^{\mu \nu }s^{\rho },
\\
X\left( s^{\nu },s^{\mu },s^{\rho }\right) &=&s^{\nu }\overline{s}^{\mu
}s^{\rho }-g^{\mu \rho }s^{\nu }+g^{\rho \nu }s^{\mu }-g^{\nu \mu }s^{\rho },
\end{eqnarray*}%
which subtracted, using the antisymmetry of $X$, yields%
\begin{equation*}
2X\left( s^{\mu },s^{\nu },s^{\rho }\right) =\left( s^{\mu }\overline{s}%
^{\nu }-s^{\nu }\overline{s}^{\mu }\right) s^{\rho }+2g^{\rho \mu }s^{\nu
}-2g^{\rho \nu }s^{\mu },
\end{equation*}%
from which it follows that%
\begin{eqnarray*}
2\left\langle X\left( s^{\mu },s^{\nu },s^{\rho }\right) ,s^{\sigma
}\right\rangle &=&\left\langle \left( s^{\mu }\overline{s}^{\nu }-s^{\nu }%
\overline{s}^{\mu }\right) s^{\rho },s^{\sigma }\right\rangle +2g^{\rho \mu
}\left\langle s^{\nu },s^{\sigma }\right\rangle -2g^{\rho \nu }\left\langle
s^{\mu },s^{\sigma }\right\rangle \\
&=&\left\langle \left( s^{\mu }\overline{s}^{\nu }-s^{\nu }\overline{s}^{\mu
}\right) s^{\rho },s^{\sigma }\right\rangle +2\left( g^{\rho \mu }g^{\nu
\sigma }-g^{\rho \nu }g^{\mu \sigma }\right) \\
&=&-\left\langle s^{\rho },\left( s^{\mu }\overline{s}^{\nu }-s^{\nu }%
\overline{s}^{\mu }\right) s^{\sigma }\right\rangle +2\left( g^{\rho \mu
}g^{\nu \sigma }-g^{\rho \nu }g^{\mu \sigma }\right) \\
&=&-4\left( \mathbf{S}^{\mu \nu }\right) ^{\rho \sigma }+2\left( \mathbf{V}%
^{\mu \nu }\right) ^{\rho \sigma },
\end{eqnarray*}%
using Eqs. (\ref{Eq:IP(xy,z)}); or, equivalently:%
\begin{equation*}
2\left( \mathbf{S}^{\mu \nu }\right) ^{\rho \sigma }=\left( \mathbf{V}^{\mu
\nu }\right) ^{\rho \sigma }-\left\langle X\left( s^{\mu },s^{\nu },s^{\rho
}\right) ,s^{\sigma }\right\rangle .
\end{equation*}%
Now, due to the complete antisymmetry of $\left\langle X\left( x,y,z\right)
,u\right\rangle $ in all its arguments, as previously established, the
second addend on the right-hand side is completely antisymmetric in $\mu \nu
\rho \sigma $. As it is by construction also a type $\left( 4,0\right) $
tensor, it must be proportional to $\varepsilon ^{\mu \nu \rho \sigma }$. By
plugging in some simple choices of $s^{\mu }$, it is readily established
that $\left\langle X\left( s^{\mu },s^{\nu },s^{\rho }\right) ,s^{\sigma
}\right\rangle =\pm \mathrm{i}\varepsilon ^{\mu \nu \rho \sigma }$, the
signs being mutually exclusive, of course. From this, the proof then follows.
\end{proof}

It should be noted that the condition $\left\langle X\left( s^{\mu },s^{\nu
},s^{\rho }\right) ,s^{\sigma }\right\rangle =\pm \mathrm{i}\varepsilon
^{\mu \nu \rho \sigma }$, for some specific choice of sign, is completely
independent of the relation $g_{\mu \nu }=\left\langle s_{\mu },s_{\nu
}\right\rangle $, Eq. (\ref{Eq:MetricDef}), posing on $s^{\mu }$ only a 
\textit{discrete} condition of 'handedness'.

\section{\label{Sec:SpinorLagrangian}The spinor Lagrangian}

Before presenting the spinor Lagrangian, first in flat spacetime, and later
in nonflat spacetime, some preliminaries: Consider the following
complex-valued type $\left( 1,1\right) $ tensor field:%
\begin{equation}
N^{\rho }{}_{\sigma }\equiv \left\langle s^{\rho },s_{\sigma }\kappa
\right\rangle ,  \label{Eq:NDef}
\end{equation}%
where $\kappa \in \mathrm{Vec}\left( \mathrm{i}\mathbb{H}\right) $ with $%
\kappa ^{2}=1$ is spacetime-independent. It is antisymmetric, $N_{\rho
\sigma }=-N_{\sigma \rho }$, because%
\begin{eqnarray}
N^{\rho }{}_{\sigma } &=&-\left\langle s^{\rho }\kappa ,s_{\sigma
}\right\rangle  \notag \\
&=&-\left\langle s_{\sigma },s^{\rho }\kappa \right\rangle  \notag \\
&=&-N_{\sigma }{}^{\rho },  \label{Eq:NAntiSym}
\end{eqnarray}%
using $\overline{\kappa }=-\kappa $, and Eqs. (\ref{Eq:IPSym}) and (\ref%
{Eq:IP(x,yz)}).

\begin{proposition}
$N^{\rho }{}_{\sigma }$ satisfies the following algebra (note the single
complex conjugation):%
\begin{eqnarray}
0 &=&M^{\mu \rho }{}_{\tau }\left( N^{\tau }{}_{\sigma }\right) ^{\ast
}+N^{\rho }{}_{\tau }M^{\mu \tau }{}_{\sigma },  \label{Eq:MN} \\
\delta _{\sigma }^{\rho } &=&N^{\rho }{}_{\tau }N^{\tau }{}_{\sigma },
\label{Eq:NN}
\end{eqnarray}%
with $M^{\mu \rho }{}_{\sigma }$ being, of course, the previously defined
tensor field, Eq. (\ref{Eq:MDef}).
\end{proposition}

\begin{proof}
By direct calculation:%
\begin{eqnarray*}
M^{\mu \rho }{}_{\tau }\left( N^{\tau }{}_{\sigma }\right) ^{\ast }+N^{\rho
}{}_{\tau }M^{\mu \tau }{}_{\sigma } &=&-\left\langle s^{\mu },s^{\rho
}s_{\tau }\right\rangle \left\langle \overline{s}^{\tau },\overline{s}%
_{\sigma }\kappa \right\rangle +\left\langle s^{\rho },s_{\tau }\kappa
\right\rangle \left\langle s^{\mu },s^{\tau }s_{\sigma }\right\rangle \\
&=&\left\langle \overline{s}^{\rho }s^{\mu },s_{\tau }\right\rangle
\left\langle s^{\tau },\kappa s_{\sigma }\right\rangle -\left\langle s^{\rho
}\kappa ,s_{\tau }\right\rangle \left\langle s^{\tau },s^{\mu }\overline{s}%
_{\sigma }\right\rangle \\
&=&\left\langle \overline{s}^{\rho }s^{\mu },\kappa s_{\sigma }\right\rangle
-\left\langle s^{\rho }\kappa ,s^{\mu }\overline{s}_{\sigma }\right\rangle \\
&=&\left\langle \overline{s}^{\rho }s^{\mu },\kappa s_{\sigma }\right\rangle
+\left\langle \overline{s}^{\mu }s^{\rho },\overline{s}_{\sigma }\kappa
\right\rangle \\
&=&\left\langle \overline{s}^{\rho }s^{\mu },\kappa s_{\sigma }\right\rangle
-\left\langle \overline{s}^{\rho }s^{\mu },\kappa s_{\sigma }\right\rangle \\
&\equiv &0,
\end{eqnarray*}%
and%
\begin{eqnarray*}
N^{\rho }{}_{\tau }N^{\tau }{}_{\sigma } &=&\left\langle s^{\rho },s_{\tau
}\kappa \right\rangle \left\langle s^{\tau },s_{\sigma }\kappa \right\rangle
\\
&=&-\left\langle s^{\rho }\kappa ,s_{\tau }\right\rangle \left\langle
s^{\tau },s_{\sigma }\kappa \right\rangle \\
&=&-\left\langle s^{\rho }\kappa ,s_{\sigma }\kappa \right\rangle \\
&=&\left\langle s^{\rho },s_{\sigma }\kappa ^{2}\right\rangle \\
&=&\left\langle s^{\rho },s_{\sigma }\right\rangle \\
&=&\delta _{\sigma }^{\rho };
\end{eqnarray*}%
using $\overline{\kappa }=-\kappa $ and $\kappa ^{\ast }=-\kappa $, and
several of the properties of the inner product listed in Sec. \ref%
{Sec:Preliminaries}.
\end{proof}

In terms of the $4\times 4$ matrix $\mathbf{N}$ with components $\mathbf{N}%
^{\rho }{}_{\sigma }\equiv N^{\rho }{}_{\sigma }$, and the previously
defined matrices $\mathbf{M}^{\mu }$, these relations may also be written
concisely in matrix notation as%
\begin{eqnarray}
\mathbf{0} &=&\mathbf{M}^{\mu }\mathbf{N}^{\ast }+\mathbf{NM}^{\mu },
\label{Eq:MN_Matrix} \\
\mathbf{1} &=&\mathbf{N}^{2},  \label{Eq:NN_Matrix}
\end{eqnarray}%
where $\mathbf{0}$ is the zero matrix.

\begin{proposition}
$\mathbf{S}^{\mu \nu }$ and $\mathbf{N}$ commute:%
\begin{equation}
\mathbf{0}=\left[ \mathbf{S}^{\mu \nu },\mathbf{N}\right] .
\label{Eq:SN_Matrix}
\end{equation}
\end{proposition}

\begin{proof}
The assertion follows immediately from the following two expressions:%
\begin{eqnarray*}
4\left( \mathbf{S}^{\mu \nu }\right) ^{\rho }{}_{\tau }\mathbf{N}^{\tau
}{}_{\sigma } &\equiv &\left\langle s^{\rho },\left( s^{\mu }\overline{s}%
^{\nu }-s^{\nu }\overline{s}^{\mu }\right) s_{\tau }\right\rangle
\left\langle s^{\tau },s_{\sigma }\kappa \right\rangle \\
&=&-\left\langle \left( s^{\mu }\overline{s}^{\nu }-s^{\nu }\overline{s}%
^{\mu }\right) s^{\rho },s_{\tau }\right\rangle \left\langle s^{\tau
},s_{\sigma }\kappa \right\rangle \\
&=&-\left\langle \left( s^{\mu }\overline{s}^{\nu }-s^{\nu }\overline{s}%
^{\mu }\right) s^{\rho },s_{\sigma }\kappa \right\rangle \\
&=&-\left\langle \overline{s}_{\sigma }\left( s^{\mu }\overline{s}^{\nu
}-s^{\nu }\overline{s}^{\mu }\right) s^{\rho },\kappa \right\rangle ,
\end{eqnarray*}%
and%
\begin{eqnarray*}
4\mathbf{N}^{\rho }{}_{\tau }\left( \mathbf{S}^{\mu \nu }\right) ^{\tau
}{}_{\sigma } &\equiv &\left\langle s^{\rho },s_{\tau }\kappa \right\rangle
\left\langle s^{\tau },\left( s^{\mu }\overline{s}^{\nu }-s^{\nu }\overline{s%
}^{\mu }\right) s_{\sigma }\right\rangle \\
&=&-\left\langle s^{\rho }\kappa ,s_{\tau }\right\rangle \left\langle
s^{\tau },\left( s^{\mu }\overline{s}^{\nu }-s^{\nu }\overline{s}^{\mu
}\right) s_{\sigma }\right\rangle \\
&=&-\left\langle s^{\rho }\kappa ,\left( s^{\mu }\overline{s}^{\nu }-s^{\nu }%
\overline{s}^{\mu }\right) s_{\sigma }\right\rangle \\
&=&-\left\langle \overline{s}^{\rho }\left( s^{\mu }\overline{s}^{\nu
}-s^{\nu }\overline{s}^{\mu }\right) s_{\sigma },\kappa \right\rangle \\
&=&-\left\langle \overline{s}_{\sigma }\left( s^{\mu }\overline{s}^{\nu
}-s^{\nu }\overline{s}^{\mu }\right) s^{\rho },\kappa \right\rangle ,
\end{eqnarray*}%
using several of the properties of the inner product listed in Sec. \ref%
{Sec:Preliminaries}, the last equality, in particular, following from Eq. (%
\ref{Eq:IPConj}) and $\overline{\kappa }=-\kappa $.
\end{proof}

\subsection{Flat spacetime}

Consider globally flat spacetime $\left\langle s_{\mu },s_{\nu
}\right\rangle =\eta _{\mu \nu }$ in Cartesian coordinates, assuming $s_{\mu
}$ to be spacetime-independent. Consider at the classical level the
following spinor Lagrangian:%
\begin{equation}
\mathcal{L}=\frac{\mathrm{i}}{2}\psi _{\rho }^{\ast }M^{\mu \rho }{}_{\sigma
}\partial _{\mu }\psi ^{\sigma }-\frac{\mathrm{i}}{2}\left( \partial _{\mu
}\psi _{\rho }\right) ^{\ast }M^{\mu \rho }{}_{\sigma }\psi ^{\sigma }+\frac{%
m}{2}\left[ \psi _{\rho }^{\ast }N^{\rho }{}_{\sigma }\psi ^{\sigma \ast
}-\psi _{\rho }\left( N^{\rho }{}_{\sigma }\right) ^{\ast }\psi ^{\sigma }%
\right] ,  \label{Eq:FlatSpacetimeL}
\end{equation}%
where, as advertised in the Introduction, the complex Grassmann-valued
spinor field $\psi $ carries a world index, rather than a standard spinor
index. Beware not to confuse $\psi ^{\rho }$, the components of that field,
with a Rarita-Schwinger/gravitino field \cite[Sec. 31.3]{Weinberg}. The
Majorana-like mass term is properly nontrivial due to Eq. (\ref{Eq:NAntiSym}%
), and it is complex self-conjugate (hermitian), by construction. The
kinetic term is complex self-conjugate (hermitian) due to Eq. (\ref%
{Eq:MComplexConj}). A Lorentz invariant Lagrangian with a Dirac-like mass
term does not seem possible. The corresponding Euler-Lagrange equations of
motion are given by%
\begin{equation}
0=\left( \mathrm{i}M^{\mu \rho }{}_{\sigma }\partial _{\mu }+mN^{\rho
}{}_{\sigma }K\right) \psi ^{\sigma },  \label{Eq:FlatSpacetimeEOM}
\end{equation}%
where $K$ is the operator of complex conjugation; or, equivalently, in
matrix/vector notation:%
\begin{eqnarray}
0 &=&\left( \mathrm{i}\mathbf{M}^{\mu }\partial _{\mu }+m\mathbf{N}K\right) 
\mathbf{\psi }  \notag \\
&\equiv &\left( \mathrm{i}\mathbf{M}^{\mu }K\partial _{\mu }+m\mathbf{N}%
\right) K\mathbf{\psi },  \label{Eq:FlatSpacetimeEOM_Matrix}
\end{eqnarray}%
where $\mathbf{\psi }$ is a four-column (vector) with components $\mathbf{%
\psi }^{\mu }=\psi ^{\mu }$.

A word of warning: Beware not to mistake $\psi _{\mu }^{\ast }$ (as in the
above Lagrangian, for instance) for being the components of the four-row
(vector) $\mathbf{\psi }^{\dagger }$, for they are not due to $\left( 
\mathbf{\psi }^{\dagger }\right) _{\mu }\equiv \left( \mathbf{\psi }\right)
^{\mu \ast }=\psi ^{\mu \ast }\neq \psi _{\mu }^{\ast }$, the inequality
being the result of the metric being indefinite; this can be contrasted with
the standard Dirac formalism where the analogous relation would read $\left( 
\mathbf{\psi }^{\dagger }\right) _{a}\equiv \left( \mathbf{\psi }\right)
^{a\ast }=\psi ^{a\ast }=\psi _{a}^{\ast }$ (now equality), because there is
no difference between having upper or lower \textit{spinor} indices.
Therefore, $\psi _{\rho }^{\ast }M^{\mu \rho }{}_{\sigma }\partial _{\mu
}\psi ^{\sigma }\neq \mathbf{\psi }^{\dagger }\mathbf{M}^{\mu }\partial
_{\mu }\mathbf{\psi }$ (inequality), for instance, which is the reason why
the above Lagrangian is given in tensor notation rather than in
matrix/vector notation. Of course, some operator, $\ddagger $ say, could be
defined so that $\psi _{\rho }^{\ast }M^{\mu \rho }{}_{\sigma }\partial
_{\mu }\psi ^{\sigma }=\mathbf{\psi }^{\ddagger }\mathbf{M}^{\mu }\partial
_{\mu }\mathbf{\psi }$ (equality), etc., but it does not seem to be a very
fruitful strategy.

Now, Eq. (\ref{Eq:FlatSpacetimeEOM_Matrix}) is readily seen to be
'Klein-Gordon compatible' in the sense that a plane wave solution $\mathbf{%
\psi }=\mathbf{\psi }_{0}\exp \left( -\mathrm{i}p\cdot x\right) $ can be
(but need not be) a solution to it only if it is on mass shell, $p^{2}=m^{2}$%
:%
\begin{eqnarray*}
0 &=&\left( \mathrm{i}\mathbf{M}^{\mu }K\partial _{\mu }+m\mathbf{N}\right)
\left( \mathrm{i}\mathbf{M}^{\nu }K\partial _{\nu }+m\mathbf{N}\right) \\
&=&\mathbf{M}^{\mu }\mathbf{M}^{\nu \ast }\partial _{\mu }\partial _{\nu }+%
\mathrm{i}m\left( \mathbf{M}^{\mu }\mathbf{N}^{\ast }+\mathbf{NM}^{\mu
}\right) K\partial _{\mu }+m^{2}\mathbf{N}^{2} \\
&=&g^{\mu \nu }\partial _{\mu }\partial _{\nu }+m^{2}\mathbf{1},
\end{eqnarray*}%
using Eq. (\ref{Eq:MMStar_Matrix}) and Eqs. (\ref{Eq:MN_Matrix})-(\ref%
{Eq:NN_Matrix}). Note that the assumed spacetime-independency of $s_{\mu }$
is used to freely move derivatives through $\mathbf{M}^{\mu }$ and $\mathbf{N%
}$.

But Klein-Gordon compatibility of the equations of motion is of course not
near sufficient to have a sensible theory. The Lagrangian must also at least
be globally Lorentz invariant, a subject to which is now turned: Assume that
under a global infinitesimal 'Lorentz' transformation, $\psi ^{\rho }$ and $%
s^{\rho }$ transform as%
\begin{eqnarray}
\delta \psi ^{\rho } &=&-\frac{1}{2}\left( d\theta _{\alpha \beta }\right)
\left( \mathbf{S}^{\alpha \beta }\right) ^{\rho }{}_{\sigma }\psi ^{\sigma },
\label{Eq:psiULorentzTrans} \\
\delta s^{\rho } &=&-\frac{1}{2}\left( d\theta _{\alpha \beta }\right)
\left( \mathbf{V}^{\alpha \beta }\right) ^{\rho }{}_{\sigma }s^{\sigma
}\equiv -\left( d\theta ^{\rho }{}_{\sigma }\right) s^{\sigma },
\label{Eq:sULorentzTrans}
\end{eqnarray}%
with $\left( \mathbf{S}^{\alpha \beta }\right) ^{\rho }{}_{\sigma }$ and $%
\left( \mathbf{V}^{\alpha \beta }\right) ^{\rho }{}_{\sigma }$ as previously
defined. They are equivalent to%
\begin{eqnarray}
\delta \psi _{\rho } &=&\frac{1}{2}\left( d\theta _{\alpha \beta }\right)
\left( \mathbf{S}^{\alpha \beta }\right) ^{\sigma }{}_{\rho }\psi _{\sigma },
\label{Eq:psiLLorentzTrans} \\
\delta s_{\rho } &=&\frac{1}{2}\left( d\theta _{\alpha \beta }\right) \left( 
\mathbf{V}^{\alpha \beta }\right) ^{\sigma }{}_{\rho }s_{\sigma }\equiv
\left( d\theta ^{\sigma }{}_{\rho }\right) s_{\sigma },
\label{Eq:sLLorentzTrans}
\end{eqnarray}%
due to Eqs. (\ref{Eq:gSg})-(\ref{Eq:gVg}). The infinitesimal parameters $%
d\theta _{\alpha \beta }=-d\theta _{\beta \alpha }\in \mathbb{R}$ are
assumed to be spacetime-independent, of course, as befits a global
transformation (in flat spacetime in Cartesian coordinates). The overall
sign of $d\theta _{\alpha \beta }$ has been chosen with foresight to have
the standard vierbein to be introduced as an auxiliary/calculational device
in Sec. \ref{Sec:ContactWithGR} transform standardly under Lorentz
transformations, compare Eq. (\ref{Eq:deltaVierbeinStd}). Due to the
defining Eqs. (\ref{Eq:MDef}) and (\ref{Eq:NDef}), the variations $\delta
s^{\rho }$ and $\delta s_{\rho }$ induce the following variations:%
\begin{eqnarray}
\delta M^{\mu \rho }{}_{\sigma } &=&-\frac{1}{2}\left( d\theta _{\alpha
\beta }\right) \left[ \left( \mathbf{V}^{\alpha \beta }\right) ^{\mu
}{}_{\tau }M^{\tau \rho }{}_{\sigma }+\left( \mathbf{V}^{\alpha \beta
}\right) ^{\rho }{}_{\tau }M^{\mu \tau }{}_{\sigma }-\left( \mathbf{V}%
^{\alpha \beta }\right) ^{\tau }{}_{\sigma }M^{\mu \rho }{}_{\tau }\right] ,
\label{Eq:MTrans} \\
\delta N^{\rho }{}_{\sigma } &=&-\frac{1}{2}\left( d\theta _{\alpha \beta
}\right) \left[ \left( \mathbf{V}^{\alpha \beta }\right) ^{\rho }{}_{\tau
}N^{\tau }{}_{\sigma }-\left( \mathbf{V}^{\alpha \beta }\right) ^{\tau
}{}_{\sigma }N^{\rho }{}_{\tau }\right] .  \label{Eq:NTrans}
\end{eqnarray}%
Substituting Eqs. (\ref{Eq:psiULorentzTrans}), (\ref{Eq:psiLLorentzTrans})
and (\ref{Eq:MTrans})-(\ref{Eq:NTrans}) into $\delta \mathcal{L}$, it is
readily seen that $\delta \mathcal{L}=0$ if and only if the following
conditions are satisfied:%
\begin{eqnarray*}
0 &=&\left[ \left( \mathbf{S}^{\alpha \beta }\right) ^{\rho }{}_{\tau
}-\left( \mathbf{V}^{\alpha \beta }\right) ^{\rho }{}_{\tau }\right] ^{\ast
}N^{\tau }{}_{\sigma }-\left[ \left( \mathbf{S}^{\alpha \beta }\right)
^{\tau }{}_{\sigma }-\left( \mathbf{V}^{\alpha \beta }\right) ^{\tau
}{}_{\sigma }\right] ^{\ast }N^{\rho }{}_{\tau }, \\
0 &=&\left[ \left( \mathbf{S}^{\alpha \beta }\right) ^{\rho }{}_{\tau
}-\left( \mathbf{V}^{\alpha \beta }\right) ^{\rho }{}_{\tau }\right] ^{\ast
}M^{\mu \tau }{}_{\sigma }-\left[ \left( \mathbf{S}^{\alpha \beta }\right)
^{\tau }{}_{\sigma }-\left( \mathbf{V}^{\alpha \beta }\right) ^{\tau
}{}_{\sigma }\right] M^{\mu \rho }{}_{\tau }-\left( \mathbf{V}^{\alpha \beta
}\right) ^{\mu }{}_{\tau }M^{\tau \rho }{}_{\sigma };
\end{eqnarray*}%
or, equivalently:%
\begin{eqnarray}
0 &=&\left( \mathbf{S}^{\alpha \beta }\right) ^{\rho }{}_{\tau }N^{\tau
}{}_{\sigma }-\left( \mathbf{S}^{\alpha \beta }\right) ^{\tau }{}_{\sigma
}N^{\rho }{}_{\tau },  \label{Eq:LorentzCondN} \\
0 &=&\left( \mathbf{S}^{\alpha \beta }\right) ^{\rho }{}_{\tau }M^{\mu \tau
}{}_{\sigma }-\left[ \left( \mathbf{S}^{\alpha \beta }\right) ^{\tau
}{}_{\sigma }\right] ^{\ast }M^{\mu \rho }{}_{\tau }+\left( \mathbf{V}%
^{\alpha \beta }\right) ^{\mu }{}_{\tau }M^{\tau \rho }{}_{\sigma },
\label{Eq:LorentzCondM}
\end{eqnarray}%
using Eq. (\ref{Eq:SVLeviCivita}) and its complex conjugate, as well as the
real-valuedness of $\mathbf{V}^{\alpha \beta }$; or, equivalently:%
\begin{eqnarray}
\mathbf{0} &=&\mathbf{S}^{\alpha \beta }\mathbf{N}-\mathbf{NS}^{\alpha \beta
}\equiv \left[ \mathbf{S}^{\alpha \beta },\mathbf{N}\right] ,
\label{Eq:LorentzCondN_Matrix} \\
\mathbf{0} &=&\mathbf{S}^{\alpha \beta }\mathbf{M}^{\mu }-\mathbf{M}^{\mu }%
\mathbf{S}^{\alpha \beta \ast }+\left( \mathbf{V}^{\alpha \beta }\right)
^{\mu }{}_{\tau }\mathbf{M}^{\tau }  \notag \\
&\equiv &\mathbf{S}^{\alpha \beta }\mathbf{M}^{\mu }-\mathbf{M}^{\mu }%
\mathbf{S}^{\alpha \beta \ast }+\left( g^{\alpha \mu }\delta _{\tau }^{\beta
}-g^{\beta \mu }\delta _{\tau }^{\alpha }\right) \mathbf{M}^{\tau }  \notag
\\
&=&\mathbf{S}^{\alpha \beta }\mathbf{M}^{\mu }-\mathbf{M}^{\mu }\mathbf{S}%
^{\alpha \beta \ast }+g^{\mu \alpha }\mathbf{M}^{\beta }-g^{\mu \beta }%
\mathbf{M}^{\alpha }.  \label{Eq:LorentzCondM_Matrix}
\end{eqnarray}%
These conditions are indeed satisfied due to Eqs. (\ref{Eq:DoubleCover}) and
(\ref{Eq:SN_Matrix}), thus proving that $\mathcal{L}$ is globally 'Lorentz'
invariant (the use of quotation marks here and previously is as a reminder
that the generators obey the 'modified Lorentz algebra', Eq. (\ref%
{Eq:ModLorentzAlgebra})).

\subsection{Nonflat spacetime}

In going to nonflat spacetime (or employing curvilinear coordinates in flat
spacetime, for that matter), the previously given Lagrangian, Eq. (\ref%
{Eq:FlatSpacetimeL}), must be generalized to%
\begin{equation}
\mathcal{L}=\frac{\mathrm{i}}{2}\psi _{\rho }^{\ast }M^{\mu \rho }{}_{\sigma
}\nabla _{\mu }\psi ^{\sigma }-\frac{\mathrm{i}}{2}\left( \nabla _{\mu }\psi
_{\rho }\right) ^{\ast }M^{\mu \rho }{}_{\sigma }\psi ^{\sigma }+\frac{m}{2}%
\left[ \psi _{\rho }^{\ast }N^{\rho }{}_{\sigma }\psi ^{\sigma \ast }-\psi
_{\rho }\left( N^{\rho }{}_{\sigma }\right) ^{\ast }\psi ^{\sigma }\right] .
\label{Eq:NonFlatSpacetimeL}
\end{equation}%
In order to retain the previously derived 'Klein-Gordon compatibility' of
the equations of motion (for the case of flat spacetime in Cartesian
coordinates), it is necessary that $\nabla _{\nu }M^{\mu \rho }{}_{\sigma
}=0 $ and $\nabla _{\nu }N^{\rho }{}_{\sigma }=0$, i.e., that $M^{\mu \rho
}{}_{\sigma }$ and $N^{\rho }{}_{\sigma }$ are covariantly constant. In view
of Eqs. (\ref{Eq:MDef}) and (\ref{Eq:NDef}), this will certainly be the case
if $s_{\mu }$ itself is identically covariantly constant:%
\begin{equation}
0\equiv \nabla _{\nu }s_{\mu }\equiv \partial _{\nu }s_{\mu }-\Gamma ^{\rho
}{}_{\mu \nu }s_{\rho },  \label{Eq:sCovConst}
\end{equation}%
a condition that is uniquely satisfied by%
\begin{equation}
\Gamma ^{\rho }{}_{\mu \nu }=\left\langle s^{\rho },\partial _{\nu }s_{\mu
}\right\rangle ,  \label{Eq:CartanConn}
\end{equation}%
the uniqueness being due to $s_{\mu }$ being a basis of $\mathbb{C}\otimes 
\mathbb{H}$, i.e., $\mathbb{C}\otimes \mathbb{H}=\mathrm{Span}_{\mathbb{C}%
}\left( s_{\mu }\right) $, as noted in Sec. \ref{Sec:Preliminaries}. These
connection coefficients are real-valued (as they should be in order to be
sensible in the realm of a Riemann-Cartan spacetime):%
\begin{eqnarray*}
\left( \Gamma ^{\rho }{}_{\mu \nu }\right) ^{\ast } &=&\left\langle 
\overline{s}^{\rho },\partial _{\nu }\overline{s}_{\mu }\right\rangle \\
&=&\left\langle s^{\rho },\partial _{\nu }s_{\mu }\right\rangle \\
&=&\Gamma ^{\rho }{}_{\mu \nu },
\end{eqnarray*}%
using $s_{\mu }^{\ast }=\overline{s}_{\mu }$ and Eq. (\ref{Eq:IPConj}). The
covariant derivative of the components of the spinor field are then
obviously taken to be%
\begin{equation}
\nabla _{\nu }\psi ^{\rho }\equiv \partial _{\nu }\psi ^{\rho }+\Gamma
^{\rho }{}_{\mu \nu }\psi ^{\mu },  \label{Eq:psiCovDer}
\end{equation}%
as alluded to in the Introduction. Note that $\nabla _{\nu }s_{\mu }\equiv 0$
does not only imply covariant constancy of $M^{\mu \rho }{}_{\sigma }$ and $%
N^{\rho }{}_{\sigma }$, but in view of Eqs. (\ref{Eq:SDef}) and (\ref%
{Eq:VDef}) also, quite satisfactorily, covariant constancy of $S^{\alpha
\beta \rho }{}_{\sigma }$ and $V^{\alpha \beta \rho }{}_{\sigma }$ (and thus
as well of the corresponding 'Lorentz' generators):%
\begin{eqnarray}
0 &\equiv &\nabla _{\nu }S^{\alpha \beta \rho }{}_{\sigma },
\label{Eq:SCovConst} \\
0 &\equiv &\nabla _{\nu }V^{\alpha \beta \rho }{}_{\sigma },
\label{Eq:VCovConst}
\end{eqnarray}%
Concerning global 'Lorentz' invariance: Assume that the now generically
spacetime-dependent parameters $d\theta _{\alpha \beta }$, figuring in Eqs. (%
\ref{Eq:psiULorentzTrans})-(\ref{Eq:sULorentzTrans}), are covariantly
constant:%
\begin{equation}
0=\nabla _{\nu }d\theta _{\alpha \beta }.  \label{Eq:dthetaCovConst}
\end{equation}%
(This condition does of course not alter any of the previous derivations for
flat spacetime in Cartesian coordinates concerning invariance under global
'Lorentz' transformations.) It implies that the connection coefficients, Eq.
(\ref{Eq:CartanConn}), are globally 'Lorentz' invariant:%
\begin{eqnarray*}
\delta \Gamma ^{\rho }{}_{\mu \nu } &\equiv &\left\langle \delta s^{\rho
},\partial _{\nu }s_{\mu }\right\rangle +\left\langle s^{\rho },\partial
_{\nu }\delta s_{\mu }\right\rangle \\
&=&-\left( d\theta ^{\rho }{}_{\sigma }\right) \left\langle s^{\sigma
},\partial _{\nu }s_{\mu }\right\rangle +\left\langle s^{\rho },\partial
_{\nu }\left[ \left( d\theta ^{\sigma }{}_{\mu }\right) s_{\sigma }\right]
\right\rangle \\
&=&-\left( d\theta ^{\rho }{}_{\sigma }\right) \left\langle s^{\sigma
},\partial _{\nu }s_{\mu }\right\rangle +\left( d\theta ^{\sigma }{}_{\mu
}\right) \left\langle s^{\rho },\partial _{\nu }s_{\sigma }\right\rangle
+\left( \partial _{\nu }d\theta ^{\sigma }{}_{\mu }\right) \left\langle
s^{\rho },s_{\sigma }\right\rangle \\
&=&\partial _{\nu }d\theta ^{\rho }{}_{\mu }+\Gamma ^{\rho }{}_{\sigma \nu
}d\theta ^{\sigma }{}_{\mu }-\Gamma ^{\sigma }{}_{\mu \nu }d\theta ^{\rho
}{}_{\sigma } \\
&\equiv &\nabla _{\nu }d\theta ^{\rho }{}_{\mu }.
\end{eqnarray*}%
This in turn implies that $\delta \nabla _{\mu }=0$, i.e., that $\nabla
_{\mu }$ is globally 'Lorentz' invariant. Therefore,%
\begin{eqnarray*}
\delta \left( \nabla _{\mu }\psi ^{\rho }\right) &=&\nabla _{\mu }\delta
\psi ^{\rho } \\
&=&-\frac{1}{2}\nabla _{\mu }\left[ \left( d\theta _{\alpha \beta }\right)
\left( \mathbf{S}^{\alpha \beta }\right) ^{\rho }{}_{\sigma }\psi ^{\sigma }%
\right] \\
&=&-\frac{1}{2}\left( d\theta _{\alpha \beta }\right) \left( \mathbf{S}%
^{\alpha \beta }\right) ^{\rho }{}_{\sigma }\nabla _{\mu }\psi ^{\sigma },
\end{eqnarray*}%
using Eqs. (\ref{Eq:SCovConst}) and (\ref{Eq:dthetaCovConst}), i.e., $\nabla
_{\mu }\psi ^{\rho }$ transforms under global 'Lorentz' transformations as $%
\psi ^{\rho }$ itself does. From this, it readily follows that the above
Lagrangian is itself globally 'Lorentz' invariant. In conjunction with
global 'Lorentz' invariance of the metric (note that this holds for
arbitrary $d\theta _{\alpha \beta }=-d\theta _{\beta \alpha }$):%
\begin{eqnarray*}
\delta \left\langle s_{\mu },s_{\nu }\right\rangle &\equiv &\left\langle
\delta s_{\mu },s_{\nu }\right\rangle +\left\langle s_{\mu },\delta s_{\nu
}\right\rangle \\
&=&\left( d\theta ^{\rho }{}_{\mu }\right) \left\langle s_{\rho },s_{\nu
}\right\rangle +\left( d\theta ^{\rho }{}_{\nu }\right) \left\langle s_{\mu
},s_{\rho }\right\rangle \\
&=&d\theta _{\nu \mu }+d\theta _{\mu \nu } \\
&\equiv &0,
\end{eqnarray*}%
this implies that the action $S=\int \mathcal{L}\sqrt{-g}d^{4}x$ is globally
'Lorentz' invariant. By construction it is, of course, also coordinate
invariant. Note that the only conditions posed are the ones of covariant
constancy of $s_{\mu }$ and $d\theta _{\alpha \beta }$, Eqs. (\ref%
{Eq:sCovConst}) and (\ref{Eq:dthetaCovConst}).

\section{\label{Sec:ContactWithGR}Making contact with general relativity}

Let $s^{a}\in \left( \mathbb{C}\otimes \mathbb{H}\right) ^{+}$ be a
spacetime-independent basis for $\mathbb{C}\otimes \mathbb{H}$ for which $%
\left\langle s^{a},s^{b}\right\rangle =\eta ^{ab}$, and consider the
expansion $s^{\mu }=e^{\mu }{}_{a}s^{a}$, where $e^{\mu }{}_{a}=\left\langle
s^{\mu },s_{a}\right\rangle \in \mathbb{R}$. Inserting this expansion into
the expression for $\delta s^{\rho }$, Eq. (\ref{Eq:sULorentzTrans}),
implies that%
\begin{eqnarray}
\delta e^{\rho }{}_{c} &=&-\frac{1}{2}\left( d\theta _{\alpha \beta }\right)
\left( \mathbf{V}^{\alpha \beta }\right) ^{\rho }{}_{\sigma }e^{\sigma
}{}_{c}  \notag \\
&=&-\frac{1}{2}\left( d\theta _{ab}\right) \left( \mathbf{V}^{ab}\right)
^{d}{}_{c}e^{\rho }{}_{d},  \label{Eq:deltaVierbeinStd}
\end{eqnarray}%
using the identity $\left( \mathbf{V}^{\alpha \beta }\right) ^{\rho
}{}_{\sigma }=e^{\alpha }{}_{a}e^{\beta }{}_{b}e^{\rho
}{}_{r}e^{s}{}_{\sigma }\left( \mathbf{V}^{ab}\right) ^{r}{}_{s}$ with $%
\left( \mathbf{V}^{ab}\right) ^{c}{}_{d}\equiv \eta ^{ac}\delta
_{d}^{b}-\eta ^{bc}\delta _{d}^{a}$ the standard vector representation of
the Lorentz algebra, and introducing $d\theta _{ab}\equiv e^{\alpha
}{}_{a}e^{\beta }{}_{b}d\theta _{\alpha \beta }$. This is the standard
infinitesimal Lorentz transformation of a vierbein, with transformation
parameters $d\theta _{ab}$. The expansion coefficients $e^{\mu }{}_{a}$ in
conjunction thus seem identifiable with the standard vierbein. In terms of
this vierbein, the previously introduced condition of covariant constancy of 
$s_{\mu }$, Eq. (\ref{Eq:sCovConst}), becomes%
\begin{equation}
\nabla _{\nu }e^{\mu }{}_{a}\equiv 0,  \label{Eq:eCovConst}
\end{equation}%
i.e., the covariant constancy of the vierbein, and the previously introduced
connection coefficients, Eq. (\ref{Eq:CartanConn}), become%
\begin{eqnarray}
\Gamma ^{\rho }{}_{\mu \nu } &=&\left\langle e^{\rho }{}_{a}s^{a},\partial
_{\nu }\left( e^{b}{}_{\mu }s_{b}\right) \right\rangle  \notag \\
&=&e^{\rho }{}_{a}\left( \partial _{\nu }e^{b}{}_{\mu }\right) \left\langle
s^{a},s_{b}\right\rangle  \notag \\
&=&e^{\rho }{}_{a}\partial _{\nu }e^{a}{}_{\mu },
\label{Eq:CartanConnVierbein}
\end{eqnarray}%
i.e., the standard Cartan connection, carrying torsion, but no curvature,
the setting thus being that of Weitzenb\"{o}ck spacetime and
teleparallelism. The spinor Lagrangian, Eq. (\ref{Eq:NonFlatSpacetimeL}),
and the formalism associated with it thus seems to be consistent with the
teleparallel formulation of general relativity. A few comments:

\begin{itemize}
\item In conjunction with Eq. (\ref{Eq:dthetaCovConst}), the covariant
constancy of the vierbein, $\nabla _{\nu }e^{\mu }{}_{a}\equiv 0$,
immediately implies that $\partial _{\nu }d\theta _{ab}\equiv \nabla _{\nu
}d\theta _{ab}\equiv 0$, i.e., that the infinitesimal parameters $d\theta
_{ab}$ are spacetime-independent, as befits a global Lorentz transformation.
Remember that for teleparallelism, only \textit{global} Lorentz
transformations are relevant, the local degrees of freedom being frozen out.

\item The spinor Lagrangian, Eq. (\ref{Eq:NonFlatSpacetimeL}), may of course
be expanded in terms of the vierbein, simply by inserting $s^{\mu }=e^{\mu
}{}_{a}s^{a}$ into the definitions of $M^{\mu \rho }{}_{\sigma }$ and $%
N^{\rho }{}_{\sigma }$, Eqs. (\ref{Eq:MDef}) and (\ref{Eq:NDef}), but in
view of the very aim of this article, that would obviously be
counterproductive as it would reintroduce Lorentz indices into the formalism.

\item The introduction and use in this section of the vierbein is \textit{not%
} tantamount to introducing Lorentz indices back into the formalism. Its
sole purpose is to be an auxiliary/calculational device for establishing the
consistency between the developed spinor Lagrangian formalism and the
teleparallel formulation of general relativity. Note that although the
teleparallel formulation of general relativity is build from the vierbein
and its first order derivatives, its Lagrangian effectively depends only on
the torsion tensor field, $T^{\rho }{}_{\mu \nu }$, compare for instance 
\cite{Maluf}, which does not carry any Lorentz indices.
\end{itemize}

\noindent So, in conclusion, collecting everything, it seems that by
combining the above spinor Lagrangian, Eq. (\ref{Eq:NonFlatSpacetimeL}),
with the Lagrangian for the teleparallel formulation of general relativity,
a formalism for the coupling of spinor fields to the gravitational field
using only world indices is provided, as asserted.

\end{document}